\documentclass[a4paper,11pt]{article}
\usepackage[utf8x,utf8]{inputenc} 
\usepackage[T1]{fontenc}
\usepackage{lmodern}
\usepackage[vmargin=33mm,hmargin=37mm]{geometry}

\usepackage{url,graphicx,color,microtype,enumitem,hyperref}
\usepackage{amsfonts,amssymb,amsmath,euscript,amsthm}

\hypersetup{
	bookmarks,
	pdftex,
    pdfauthor={Fedor V. Fomin, Petr A. Golovach and Janne H. Korhonen},
    pdftitle={On the parameterized complexity of cutting a few vertices from a graph},
	colorlinks=true,
	linkcolor=black,
	citecolor=black,
	filecolor=black,
	urlcolor=black,
}

\setitemize{noitemsep}
\setenumerate{noitemsep,itemsep=1ex}

\newtheorem{theorem}{Theorem}
\newtheorem{lemma}[theorem]{Lemma}
\theoremstyle{definition}
\newtheorem{definition}[theorem]{Definition}
\theoremstyle{remark}

\newcommand{\Z}{\mathbb{Z}}
\newcommand{\X}{\EuScript{X}}

\newcommand{\card}[1]{\left\lvert {#1} \right\rvert}

\DeclareMathOperator{\operatorClassNP}{NP}
\newcommand{\classNP}{\ensuremath{\operatorClassNP}}
\DeclareMathOperator{\operatorClassFPT}{FPT}
\newcommand{\classFPT}{\ensuremath{\operatorClassFPT}}
\DeclareMathOperator{\operatorClassW}{W}
\newcommand{\classW}[1]{\ensuremath{\operatorClassW[#1]}}
\DeclareMathOperator{\operatorClassXP}{XP}
\newcommand{\classXP}{\ensuremath{\operatorClassXP}}

\newcommand{\vbound}[1]{\card{N(#1)}}
\newcommand{\vboundG}[2]{\card{N_{#1}(#2)}}
\newcommand{\ebound}[1]{\card{\partial(#1)}}

\newcommand{\defparproblem}[5]{
  \vspace{1mm}
\noindent\fbox{
  \begin{minipage}{0.97\textwidth}
  \begin{tabular*}{\textwidth}{@{\extracolsep{\fill}}lr} #1  \\ \end{tabular*}
  {\bf{Input:}} #2  \\
    {\bf{Parameter 1:}} #3   \\
    {\bf{Parameter 2:}} #4\\
  {\bf{Question:}} #5
  \end{minipage}
  }
  \vspace{1mm}
}

\newenvironment{myabstract}
               {\list{}{\listparindent 1.5em%
                        \itemindent    \listparindent
                        \leftmargin    0pt
                        \rightmargin   0pt
                        \parsep        0pt}%
                \item\relax}
               {\endlist}

\newenvironment{mycover}
               {\list{}{\listparindent 0pt
                        \itemindent    \listparindent
                        \leftmargin    0pt
                        \rightmargin   0pt
                        \parsep        0pt}%
                \raggedright
                \item\relax}
               {\endlist}

\begin{document}

\vspace*{2ex}
\begin{mycover}
{\LARGE \textbf{On the parameterized complexity of cutting a few vertices from a graph\footnote{This work is supported by the European Research Council (ERC) via grant Rigorous Theory of Preprocessing, reference 267959 (F.F., P.G.) and by the Helsinki Doctoral Programme in Computer Science -- Advanced Computing and Intelligent Systems (J.K.).}${}^,$\footnote{This work has been published as a part of the proceedings of the 38th International Symposium on Mathematical Foundations of Computer Science ({MFCS} 2013) \cite{small-cuts-conf}. The final publication is available at \url{http://link.springer.com/chapter/10.1007\%2F978-3-642-40313-2_38}.}}\par}

\bigskip
\bigskip

\medskip
\textbf{Fedor V.\ Fomin}\\
{\small Department of Informatics, University of Bergen, Norway\\
\nolinkurl{fedor.fomin@ii.uib.no}\par}

\medskip
\textbf{Petr A.\ Golovach}\\
{\small Department of Informatics, University of Bergen, Norway\\
\nolinkurl{petr.golovach@ii.uib.no}\par}

\medskip
\textbf{Janne H.\ Korhonen}\\
{\small Helsinki Institute for Information Technology HIIT, \\
Department of Computer Science, University of Helsinki, Finland\\
\nolinkurl{janne.h.korhonen@cs.helsinki.fi}\par}
\end{mycover}

\bigskip
\bigskip

\begin{myabstract}
\noindent\textbf{Abstract.} We study the parameterized complexity of separating a small set of vertices from a graph by a small vertex-separator. That is, given a graph $G$ and integers $k$, $t$, the task is to find a vertex set $X$ with $|X| \le k$ and $|N(X)| \le t$. We show that
\begin{itemize}
	\item the problem is fixed-parameter tractable (FPT) when parameterized by $t$ but W[1]-hard when parameterized by $k$, and
	\item a terminal variant of the problem, where $X$ must contain a given vertex $s$, is W[1]-hard when parameterized either by $k$ or by $t$ alone, but is FPT when parameterized by $k + t$.  
\end{itemize}
We also show that if we consider edge cuts instead of vertex cuts, the terminal variant is NP-hard.
\end{myabstract}

\thispagestyle{empty}
\setcounter{page}{0}
\newpage

% --- MAIN MATTER ---

% --- INTRODUCTION ---

\section{Introduction}

We investigate two related problems that concern separating a small vertex set from a graph $G = (V,E)$. Specifically, we consider finding a vertex set $X$ of size at most $k$ such that
\begin{enumerate}
	\item $X$ is separated from the rest of $V$ by a small cut (e.g. \emph{finding communities in a social network}, cf.~\cite{li2012}), or
	\item $X$ is separated from the rest of $V$ by a small cut and contains a specified terminal vertex $s$ (e.g. \emph{isolating a dangerous node}, cf.~\cite{galbiati2011,hayrapetyan2005}).
\end{enumerate}
We focus on \emph{parameterized complexity} of the \emph{vertex-cut versions} of these problems.

\paragraph{Parameterized vertex cuts.} Our interest in the vertex-cut version stems from the following parameterized separation problem, studied by Marx \cite{marx2006parameterized}. Let $N(X)$ denote the vertex-neighborhood of $X$.

\medskip 

\defparproblem{\textsc{Cutting $k$ Vertices}}{Graph $G=(V,E)$, 
integers $k\geq 1$, $t\geq 0$
}{$k$}{$t$}{Is there a set $X \subseteq V$ such that $\card{X}=k$ and $\card{N(X)} \leq t$?}

\medskip

In particular, Marx showed that \textsc{Cutting $k$ Vertices} is \classW{1}-hard even when parameterized by both $k$ and $t$. We contrast this result by investigating the parameterized complexity of the two related separation problems with relaxed requirement on the size of the separated set $X$.

\medskip

\defparproblem{\textsc{Cutting at Most  $k$ Vertices}}{Graph $G=(V,E)$, 
integers $k\geq 1$, $t\geq 0$}{$k$}{$t$}{Is there a non-empty set $X \subseteq V$ such that $\card{X} \le k$ and $\card{N(X)} \le t$?}

\medskip

\defparproblem{\textsc{Cutting at Most  $k$ Vertices with Terminal}}{Graph $G=(V,E)$, terminal vertex $s$, 
integers $k\geq 1$, $t\geq 0$
}{$k$}{$t$}{Is there a non-empty set $X \subseteq V$ such that $s \in X$, $\card{X} \le k$ and $\card{N(X)} \le t$?}

\medskip

We show that these closely related problems exhibit quite different complexity behaviors. In particular, we show that \textsc{Cutting at Most  $k$ Vertices} is fixed-parameter tractable (\classFPT) when parameterized by the size of the separator $t$, while we need both $k$ and $t$ as parameters to obtain an \classFPT{} algorithm for \textsc{Cutting at Most  $k$ Vertices with Terminal}. A full summary of the parameterized complexity of these problems and our results is given in Table \ref{tabl:compl}. 

\begin{table}[b]
\begin{center}
\begin{tabular}{c|c|c|c}
Parameter & \textsc{Cutting $k$} & \textsc{Cutting  $\leq k$} & \textsc{Cutting    $\leq k$ Vertices}  \\
 &  \textsc{Vertices}  & \textsc{Vertices} & \textsc{with Terminal} \\
\hline
$k$ &  \classW{1}-hard,   \cite{marx2006parameterized} &  \classW{1}-hard,  Thm~\ref{thm:svc_is_hard}&    \classW{1}-hard, Thm~\ref{thm:svc_is_hard} \\
\hline
$t$ & \classW{1}-hard,    \cite{marx2006parameterized} & \classFPT,  Thm~\ref{thm:svc_fpt}  &    \classW{1}-hard,  Thm~\ref{thm:svc_terminal_hard}  \\
\hline
 $k$ and $t$ & \classW{1}-hard,    \cite{marx2006parameterized}  & \classFPT, Thm~\ref{thm:svc_fpt_kt}  &  \classFPT,     Thm~\ref{thm:svc_fpt_kt} \\
\end{tabular}
\vspace{2mm}
\caption{Parameterized complexity of  \textsc{Cutting $k$ Vertices}, \textsc{Cutting at Most  $k$ Vertices},  and \textsc{Cutting at Most  $k$ Vertices with Terminal}.}\label{tabl:compl}
\end{center}
\end{table}

The main algorithmic contribution of our paper is the proof that  \textsc{Cutting at most  $k$ vertices}  is \classFPT{} when parameterized by $t$ (Theorem \ref{thm:svc_fpt}). To obtain this result, we utilize the concept of \emph{important separators} introduced by Marx \cite{marx2006parameterized}. However, a direct application of important separators---guess a vertex contained in the separated set, and find a minimal set containing this vertex that can be separated from the remaining graph by at most $t$ vertices---does not work. Indeed, pursuing this approach would bring us to essentially solving  \textsc{Cutting at most  $k$ vertices with terminal}, which is
\classW{1}-hard when parameterized by $t$. 
Our \classFPT{} algorithm is based on new  structural results about unique important separators of minimum size separating pairs of vertices.   
%PG:
We also observe that it is unlikely that \textsc{Cutting at most  $k$ vertices} has a polynomial kernel.

\paragraph{Edge cuts.} Although our main focus is on vertex cuts, we will also make some remarks on the edge-cut versions of the problems. In particular, the edge-cut versions again exhibit a different kind of complexity behavior. Let $\partial(X)$ denote the edge-boundary of $X$.

\medskip

\defparproblem{\textsc{Cutting at Most  $k$ Vertices by Edge-Cut}}{Graph $G=(V,E)$, 
integers $k\geq 1$, $t\geq 0$}{$k$}{$t$}{Is there a non-empty set $X \subseteq V$ such that $\card{X} \le k$ and $\ebound{X} \le t$?}
\medskip

\defparproblem{\textsc{Cutting $k$ Vertices by Edge-Cut with Terminal}}
{Graph $G=(V,E)$, terminal vertex $s$, 
%PG:
integers $k\geq 1$, $t\geq 0$}
{$k$}{$t$}
{Is there a set $X\subseteq V$ such that $s\in X$, $\card{X}\leq k$ and $\ebound{X}\leq t$?}

\medskip

Results by Watanabe and Nakamura~\cite{WatanabeN87} imply that \textsc{Cutting at most  $k$ vertices by edge-cut} can be done in polynomial time even when $k$ and $t$ are part of the input; more recently, Armon and Zwick \cite{armon2006} have shown that this also holds in the edge-weighted case. Lokshtanov and Marx~\cite{Lokshtanov2013278} have proven that \textsc{Cutting at most $k$ vertices by edge-cut with terminal} is fixed-parameter tractable when parameterized by $k$ or by $t$; see also \cite{yixin}. We complete the picture by showing that \textsc{Cutting at most  $k$ vertices by edge-cut with terminal} is \classNP-hard (Theorem \ref{thm:sec_terminal_np}). The color-coding techniques we employ in Theorem~\ref{thm:svc_fpt_kt} also give a simple algorithm with running time $2^{k + t + o( k + t )}\cdot n^{O(1)}$.

Related edge-cut problems have received attention in the context of approximation algorithms. In contrast to \textsc{Cutting at most  $k$ vertices by edge-cut}, finding a minimum-weight edge-cut that separates exactly $k$ vertices is NP-hard. Feige et al. \cite{Feige2003643} give a PTAS for $k = O(\log n)$ and an $O(k / \log n)$-approximation for $k = \Omega(\log n)$; Li and Zhang \cite{li2012} give an $O(\log n)$-approximation. Approximation algorithms have also been given for unbalanced $s$-$t$-cuts, where $s$ and $t$ are specified terminal vertices and the task is to find an edge cut $(X, V \setminus X)$ with $s \in X$ and $t \in V \setminus S$ such that (a) $\card{X} \le k$ and weight of the cut is minimized \cite{galbiati2011,li2012}, or (b) weight of the cut is at most $w$ and $\card{X}$ is minimized \cite{hayrapetyan2005}.

% --- DEFINITIONS AND PRELIMINARIES ---

\section{Basic definitions and preliminaries}\label{sec:defs}

\paragraph{Graph theory.} We follow the conventions of Diestel~\cite{Diestel10} with graph-theoretic notations. 
We only consider finite, undirected graphs that do not contain loops or multiple
edges.
The vertex set of a graph $G$ is denoted by $V(G)$ and  
the edge set is denoted by $E(G)$, or simply by $V$ and $E$, respectively.
Typically we use $n$ to denote the number of vertices of $G$ and $m$ the number of edges.

For a set of vertices $U\subseteq V(G)$, we write
$G[U]$ for the subgraph of $G$ induced by $U$,
and $G-U$ for the graph obtained form $G$ by the removal of all the vertices of $U$, i.e., the subgraph of $G$ induced by $V(G)\setminus U$.
Similarly, for a set of edges $A$, the graph obtained from $G$ by the removal of all the edges in $A$ is denoted by $G-A$.

For a vertex $v$, we denote by $N_G(v)$ its
(\emph{open}) \emph{neighborhood}, that is, the set of vertices which are
adjacent to $v$. The \emph{degree} of a vertex $v$ is $d_G(v)=\card{N_G(v)}$.
For a set of vertices $U\subseteq V(G)$, we write $N_G(U)=\cup_{v\in U}N_G(v)\setminus U$ and
$\partial_G(U) = \left\{uv\in E(G) \mid u\in U, v\in V(G)\setminus U\right\}$. 
We may omit subscripts in these notations if there is no danger of ambiguity.

\paragraph{Submodularity.} We will make use of the well-known fact that given a graph $G$, the mapping $2^V \to \Z$ defined by $U \mapsto \vbound{U}$ is \emph{submodular}. That is, for $A, B \subseteq V$ we have
\begin{equation}
\vbound{A\cap B}+\vbound{A\cup B} \le \vbound{A}+ \vbound{B}\,.\label{eq:vsubm}
\end{equation}

\paragraph{Important separators.} Let $G$ be a graph. For disjoint sets $X, Y \subseteq V$, a vertex set $S \subseteq V \setminus (X \cup Y)$ is a (\emph{vertex}) \emph{$(X,Y)$-separator} if there is no path from $X$ to $Y$ in $G-S$. An \emph{edge $(X,Y)$-separator} $A \subseteq E$ is defined analogously. 
Note that we do not allow deletion of vertices in $X$ and $Y$, and thus there are no vertex $(X,Y)$-separators if $X$ and $Y$ are adjacent.
As our main focus is on the vertex-cut problems, all separators are henceforth vertex-separators unless otherwise specified.

We will make use of the concept of \emph{important separators}, introduced by Marx~\cite{marx2006parameterized}.
A vertex $v$ is \emph{reachable} from a set $X\subseteq V$ if $G$ has a path that joins a vertex of $X$ and $v$. 
For any sets $S$ and $X \subseteq V \setminus S$, we denote the set of vertices reachable from $X$ in $G-S$ by $R(X,S)$.
An $(X,Y)$-separator $S$ is \emph{minimal} if no proper subset of $S$ is an $(X,Y)$-separator. For $(X,Y)$-separators $S$ and $T$, we say that $T$ \emph{dominates} $S$ if $\card{T} \le \card{S}$ and $R(X,S)$ is a proper subset of $R(X,T)$.
For singleton sets, we will write $x$ instead of $\{ x \}$ in the notations defined above.

\begin{definition}[\cite{marx2006parameterized}]
An $(X,Y)$-separator $S$ is \emph{important} if it is minimal and there is no other $(X,Y)$-separator dominating $S$.
\end{definition}

In particular, this definition implies that for any $(X,Y)$-separator $S$ there exists an important $(X,Y)$-separator $T$ with $\card{T} \le \card{S}$ and $R(X,T) \supseteq R(X,S)$. If $S$ is not important, then at least one of the aforementioned relations is proper.

The algorithmic usefulness of important separators follows from the fact that the number of important separators of size at most $t$ is bounded by $t$ alone, and furthermore, these separators can be listed efficiently. Moreover, minimum-size important separators are unique and can be found in polynomial time. That is, we will make use of the following lemmas.

\begin{lemma}[\cite{chen2009improved}]\label{lemma:computing_imp_seps}
For any disjoint sets $X, Y \subseteq V$, the number of important $(X,Y)$-separators of size at most $t$ is at most $4^{t}$, and all important $(X,Y)$-separators of size at most $t$ can be listed in time $4^{t}\cdot n^{O(1)}$.
\end{lemma}

\begin{lemma}[\cite{marx2006parameterized}]\label{lemma:unique_min_sep}
For any sets $X, Y \subseteq V$, if there exists an $(X,Y)$-separator, then there is exactly one important $(X,Y)$-separator of minimum size. This separator can be found in polynomial time.
\end{lemma}

\paragraph{Parameterized complexity.} We will briefly review the basic notions of parameterized complexity, though we refer to the books of Downey and Fellows~\cite{DowneyF99}, Flum and Grohe~\cite{FlumG06}, and   Niedermeier~\cite{Niedermeierbook06} for a detailed introduction. 
Parameterized complexity is a two-dimensional framework for studying the computational complexity of a problem; one dimension is the input size
$n$ and another one is a parameter $k$. A parameterized problem is \emph{fixed-parameter tractable} (or \classFPT) if it can be solved in time $f(k)\cdot n^{O(1)}$ for some function $f$, and in the class \classXP{} if it can be solved in time $O\left(n^{f(k)}\right)$ for some function $f$.

Between \classFPT{} and \classXP{} lies the class \classW{1}. One of basic assumptions of the parameterized complexity theory is the conjecture that $\classW{1}\neq \classFPT$, and it is thus held to be unlikely that a \classW{1}-hard problem would be in \classFPT{}. For exact definition of \classW{1}, 
%PG:
%and the related notion of parameterized reduction 
we refer to the books mentioned above.  We mention only that {\sc Indpendent Set} and {\sc Clique}
parameterized by solution size are two fundamental problems that are known to be \classW{1}-complete.

The basic way of showing that a parameterized problem is unlikely to be fixed-parameter tractable is to prove \classW{1}-hardness. To show that a problem is \classW{1}-hard, it is enough to give a \emph{parameterized reduction} from a known \classW{1}-hard problem.
That is, let $A,B$ be parameterized problems. We say that $A$ is (uniformly many-one) {\em \classFPT{}-reducible} to $B$ if there exist functions
$f,g:\mathbb{N}\rightarrow \mathbb{N}$, a constant $c \in \mathbb{N}$ and
 an algorithm $\mathcal{A}$ that transforms an instance $(x,k)$ of $A$ into an instance $(x',g(k))$ of $B$
in time $f(k) |x|^c$ so that $(x,k) \in A$ if and only if $(x',g(k)) \in B$.

\paragraph{Cutting problems with parameters $k$ and $t$.} In the remainder of this section, we consider \textsc{Cutting at most  $k$ vertices}, \textsc{Cutting at most $k$ vertices with terminal}, and \textsc{Cutting at most  $k$ vertices by edge-cut with terminal} with parameters $k$ and $t$. We first note that if there exists a solution for one of the problems, then
%PG:
% there is also a connected solution; 
there is also a solution in which $X$ is connected; 
% to see this, we observe that if $X$ is a set with $\card{X} \le k$ and $\vbound{X} \le t$, then for any maximal connected $Y \subseteq X$ we also have $\card{Y} \le k$ and $\vbound{Y} \le t$.
indeed, it suffices to take any maximal connected $Y \subseteq X$. Furthermore, we note finding a connected set $X$ with $\card{X} = k$ and $\vbound{X} \le t$ is fixed-parameter tractable with parameters $k$ and $t$ due to a result by Marx~\cite[Theorem 13]{marx2006parameterized}, and thus 
\textsc{Cutting at most $k$ vertices} is also fixed-parameter tractable with parameters $k$ and $t$.

We now give a simple color-coding algorithm~\cite{AlonYZ95,cai2006random} for the three problems with parameters $k$ and $t$, in particular improving upon the running time of the aforementioned algorithm for \textsc{Cutting at most  $k$ vertices}.

\begin{theorem}\label{thm:svc_fpt_kt}
\textsc{Cutting at most  $k$ vertices}, \textsc{Cutting at most $k$ vertices with terminal}, and \textsc{Cutting at most  $k$ vertices by edge-cut with terminal} can be solved in time $2^{k + t}\cdot (k + t)^{O(\log (k+t))}\cdot n^{O(1)}$.
\end{theorem}

\begin{proof}
We first consider a $2$-colored version of \textsc{Cutting at most $k$ vertices}. That is, we are given a graph $G$ where each vertex is either colored red or blue (this is not required to be a proper coloring), and the task is to find a connected red set $X$ with $\card{X} \le k$ such that $N(X)$ is blue and $\vbound{X} \le t$. If such a set exists, it can be found in polynomial time by trying all maximal connected red sets.

Now let $G = (V,E)$ be a graph. Assume that there is a set $X$ with $\card{X} \le k$ and $\vbound{X} \le t$; we may assume that $X$ is connected. It suffices to find a coloring of $V$ such that $X$ is colored red and $N(X)$ is colored blue. This can be done by coloring each vertex $v$ either red or blue independently and uniformly at random. Indeed, this gives a desired coloring with probability at least $2^{-(k+t)}$, which immediately yields a $2^{k+t}\cdot n^{O(1)}$ time randomized algorithm for \textsc{Cutting at Most  $k$ Vertices}.

This algorithm can be derandomized in standard fashion using universal sets (compare with Cai et al.~\cite{cai2006random}). Recall that \emph{a $(n,\ell)$-universal set} is a collection of binary vectors of length $n$ such that for each index subset of size $\ell$, each of the $2^\ell$ possible combinations of values appears in some vector of the set. A construction of Naor et al.~\cite{naor1995splitters} gives a $(n,\ell)$-universal set of size $2^\ell\cdot \ell^{O(\log \ell)} \log n$ that can be listed in linear time. It suffices to try all colorings induced by a $(n,k+t)$-universal set obtained trough this construction. 

The given algorithm works for \textsc{Cutting at most  $k$ vertices with terminal}  with obvious modifications. That is, given a coloring, we simply check if the terminal $s$ is red and its connected red component is a solution. This also works for \textsc{Cutting at most  $k$ vertices by edge-cut with terminal}, as we have $\vbound{X} \le \ebound{X}$.
\end{proof}

% --- PARAMETERIZED BY T ---

\section{Cutting at most  $k$ vertices parameterized by $t$}

In this section we show that \textsc{Cutting at Most  $k$ Vertices} is fixed-parameter tractable when parameterized by the size of the separator $t$ only. 
Specifically, we will prove the following theorem.

\begin{theorem}\label{thm:svc_fpt}
\textsc{Cutting at Most  $k$ Vertices}  can be solved in time $4^{t}\cdot n^{O(1)}$.
\end{theorem}

The remainder of this section consists of the proof of Theorem \ref{thm:svc_fpt}. Note that we may assume $ \frac{3}{4}t < k < n - t$. Indeed, if $k \le c t$ for a fixed constant $c < 1$,  then we can apply the algorithm of Theorem \ref{thm:svc_fpt_kt} to solve \textsc{Cutting at Most  $k$ Vertices} in time $4^t n^{O(1)}$. On the other hand, if $k \geq n - t$, then any vertex set $X$ of size $k$ is a solution, as $\vbound{X} \le n - k \le t$.

We start by guessing a vertex $u \in V$ that belongs to a solution set $X$ if one exists; specifically, we can try all choices of $u$. We cannot expect to necessarily find a solution $X$ that contains the chosen vertex $u$, even if the guess is correct, as the terminal variant is \classW{1}-hard. We will nonetheless try; turns out that the only thing that can prevent us from finding a solution containing $u$ is that we find a solution \emph{not} containing $u$.

With $u$ fixed, we compute for each $v \in V \setminus (\{ u \} \cup N(u) )$ the unique minimum important $(u,v)$-separator $S_v$. This can be done in polynomial time by Lemma \ref{lemma:unique_min_sep}. Let $V_0$ be set of those $v$ with $\card{S_v} \le t$, and denote $R(v) = R(v,S_v)$. Finally, let $\X$ be a set family consisting of those $R(v)$ for $v \in V_0$ that are inclusion-minimal,
%PG:
i.e., if $R(v)\in \X$, then there is no $w\in V_0$ such that $R(w) \subsetneq R(v)$. 
Note that we can compute the sets $V_0$, $R(v)$ and $\X$ in polynomial time.

There are now three possible cases that may occur.
\begin{enumerate}
    \item If $V_0=\emptyset$, we conclude that we have no solution containing $u$.
    \item If there is $v\in V_0$ such that $\card{R(v)}\leq k$, then $X=R(v)$ gives a solution, and we stop and return a YES-answer.
    \item Otherwise, $\X$ is non-empty and for all sets $A \in \X$ we have $\card{A}>k$.
\end{enumerate}
We only have to consider the last case, as otherwise we are done. We will show that in that case, the sets $A \in \X$ can be used to find a solution $X$ containing $u$ if one exists. For this, we need the following structural results about the sets $R(v)$.

\begin{lemma}\label{lemma:R_containment}
For any $v, w \in V_0$, if $w \in R(v)$ then $R(w) \subseteq R(v)$.
\end{lemma}

\begin{proof}
Let $A = R(v)$ and $B = R(w)$. Since $S_v = N(A)$ is a $(u,v)$-separator of minimum size, we must have $\vbound{A \cup B} \ge \vbound{A}$. By \eqref{eq:vsubm}, we have
\[ \vbound{A \cap B} \le \vbound{A} + \vbound{B} - \vbound{A \cup B} \le \vbound{B}\,.\]
%PG:
Because $w \in A$, the set $N(A \cap B)$ is a $(u,w)$-separator.
Thus, if $B \not= A \cap B $, then $N(A \cap B)$ is a $(u,w)$-separator that witnesses that $S_w$ is not an important separator. But this is not possible by the definition of $S_w$, so we have $B = A \cap B \subseteq A$.
\end{proof}

\begin{lemma}\label{lemma:R_disjoint}
Any distinct $A, B \in \X$ are disjoint.
\end{lemma}

\begin{proof}
Assume that $A, B \in \X$ are distinct and intersect. Then there is $v \in A \cap B$.
%PG:
Since $v\in A$, the set $N(A)$ is a $(u,v)$-separator of size at most $t$, and $v\in V_0$.
Recall that $\X$ contains inclusion-minimal sets $R(w)$ for $w\in V_0$. 
But by Lemma~\ref{lemma:R_containment},  $R(v)$ is a proper subset of both $A$ and $B$, which is not possible by the definition of $\X$.
\end{proof}

Now assume that the input graph $G$ has a solution for \textsc{Cutting at Most  $k$ Vertices} containing $u$. In particular, then there is an inclusion-minimal set $X \subseteq V$ with $u \in X$ satisfying $\card{X} \le k$ and $\vbound{X} \le t$. Let us fix one such set~$X$.

\begin{lemma}\label{lemma:R_intersects_U}
For all $A \in \X$, the set $A$ is either contained in $X \cup N(X)$ or does not intersect it.
\end{lemma}

\begin{proof}
Suppose that there is a set $A \in \X$ that intersects both $X \cup N(X)$ and its complement. Let $Y = V \setminus (X \cup N(X))$.

Now let $v \in A \cap Y$. By Lemma \ref{lemma:R_containment}, we have $R(v) = A$. 
If $\vbound{A \cup Y} > \vbound{Y}$ then it follows from \eqref{eq:vsubm} that
\[ \vbound{A \cap Y} \le \vbound{A} + \vbound{Y} - \vbound{A \cup Y} < \vbound{A}\,.\]
However, this would imply that $N(A \cap Y)$ is a $(u,v)$-separator smaller than $S_v = N(A)$.

Thus, we have $\vbound{A \cup Y} \le \vbound{Y}$. But $X' = X \setminus (A \cup Y \cup N(A \cup Y))$ is a proper subset of $X$; furthermore, any vertex of $N(X')$ that is not in $N(A \cup Y)$ is also in $N(X) \setminus N(Y)$, so we have $\vbound{X'} \le \vbound{X} \le t$. This is in contradiction with the minimality of $X$.
\end{proof}

\begin{lemma}\label{lemma:finding_solution}
Let $Z$ be the union of all $A \in \X$ that do not intersect $X \cup N(X)$. Then $Z \not= \emptyset$ and there is an important $(Z,u)$-separator $S$ of size at most $t$ such that $\card{R(u,S)} + \card{S} \le k + t$.
\end{lemma}

\begin{proof}
Let $S = N(X)$. Consider an arbitrary $v \in V \setminus (X \cup S)$; such vertex exists, since $k + t < n$. Since $S$ separates $v$ from $u$, the set $R(v)$ is well-defined.

Suppose now that $R(v)$ is not contained in $R(v,S)$. Let $B = R(u,S_v)$. Since $S_v$ is a minimum-size $(u,v)$-separator we have $\vbound{B} = \vbound{R(v)}$. But $N(X \cup B)$ also separates $u$ and $v$, so we have $\vbound{X \cup B} \ge \vbound{R(v)} = \vbound{B}$. By \eqref{eq:vsubm}, we have
\[ \vbound{X \cap B} \le \vbound{X} + \vbound{B} - \vbound{ X \cup B} \le \vbound{X} \le t\,.\]
But since $R(v)$ is not contained $R(v,S)$, it follows that $X \cap B$ is a proper subset of $X$, which contradicts the minimality of $X$.

Thus we have $R(v) \subseteq R(v,S)$. It follows that $R(v,S)$ contains a set $A \in \X$, which implies that $Z \not= \emptyset$ and $v \in R(A,S) \subseteq R(Z,S)$. Furthermore, since $v \in V \setminus (X \cup S)$ was chosen arbitrarily, we have that $R(Z,S) = V \setminus (X \cup S)$.

If $S$ is an important $(Z,u)$-separator, we are done. Otherwise, there is an important $(Z,u)$-separator $T$ with $\card{T} \le \card{S}$ and $R(Z,S) \subseteq R(Z,T)$. But then we have $\card{T} \le t$, and $R(u,T) \cup T \subseteq X \cup S$, that is, $\card{R(u,T) \cup T} \le k + t$.
\end{proof}

Recall now that we may assume $\card{A} > k$ for all $A \in \X$. Furthermore, we have $\card{X \cup N(X)} \le k + t < \left(2 + \frac{1}{3}\right) k$ and the sets $A \in \X$ are disjoint by Lemma \ref{lemma:R_disjoint}. Thus, at most two sets $A \in \X$ fit inside $X \cup N(X)$ by Lemma \ref{lemma:R_intersects_U}. This means that if we let $Z$ be the union of all $A \in \X$ that do not intersect $X \cup N(X)$, then as we have already computed $\X$, we can guess $Z$ by trying all $O(n^2)$ possible choices.

Assume now that $X$ is a minimal solution containing $u$ and our guess for $Z$ is correct. We enumerate all important $(Z,u)$-separators of size at most $t$. We will find by Lemma \ref{lemma:finding_solution} an important $(Z,u)$-separator $S$ such that $\card{S} \le t$ and $\card{R(u,S)} + \card{S} \le k + t$. If $\card{R(u,S)} \le k$, we have found a solution.
%PG:
Otherwise, we delete a set $S'$ of $\card{R(u,S)}-k$ elements from $R(u,S)$ to obtain a solution $X'$. To see that this suffices, observe that $N(X')\subseteq S'\cup S$. Therefore, $\card{N(X')}\leq \card{S'}+\card{S}=\card{R(u,S)}-k+\card{S} \leq t$. 
As all important $(Z,u)$-separators can be listed in time $4^t\cdot n^{O(1)}$ by Lemma \ref{lemma:computing_imp_seps}, the proof of Theorem \ref{thm:svc_fpt} is complete.

% \begin{remark}
% A graph can contain $\binom{n}{t}$ inclusion-minimal sets $X$ with $\card{X}\leq k$ and $\vbound{X}\leq t$. Indeed, consider a complete graph and $ t = n - k $.
% \end{remark}

% --- HARDNESS RESULTS ---

\section{Hardness results}
%PG:
We start this section by complementing Theorem~\ref{thm:svc_fpt}, as we show that \textsc{Cutting at Most  $k$ Vertices} is \classNP-complete and \classW{1}-hard when parameterized by $k$. We also show that same holds for
\textsc{Cutting at Most  $k$ Vertices with Terminal}. Note that both of these problems are in \classXP{} when parameterized by $k$, as they can be solved by checking all vertex subsets of size at most $k$.

%PG:
\begin{theorem}\label{thm:svc_is_hard}
\textsc{Cutting at Most  $k$ Vertices} and \textsc{Cutting at Most  $k$ Vertices with Terminal} are \classNP-complete and \classW{1}-hard with the parameter~$k$.
\end{theorem}

\begin{proof}
We prove the \classW{1}-hardness claim for \textsc{Cutting at Most  $k$ Vertices} by a reduction from  \textsc{Clique}. 
Recall that this \classW{1}-complete (see~\cite{DowneyF99}) problem asks for a graph $G$ and a positive integer $k$ where $k$ is a parameter, whether $G$ contains a clique of size $k$.
Let $(G,k)$ be an instance of \textsc{Clique}, $n=\card{V(G)}$ and $m=\card{E(G)}$; we construct an instance $(G',k',t)$ of \textsc{Cutting at Most  $k$ Vertices} as follows. Let $H_V$ be a clique of size $n^3$ and identify $n$ vertices of $H_V$ with the vertices of $G$. Let $H_E$ be a clique of size $m$ and identify the vertices of $H_E$ with the edges of $G$. Finally, add an edge between vertex $v$ of $H_V$ and vertex $e$ of $H_E$ whenever $v$ is incident to $e$ in $G$. Set $k' = \binom{k}{2}$ and $t = k + m - \binom{k}{2}$. The construction is shown in Fig.~\ref{fig:W} a).

If $G$ has a $k$-clique $K$, then for the set $X$ that consists of the vertices $e$ of $H_E$ corresponding to edges of $K$ we have $\card{X} = \binom{k}{2}$ and $\vboundG{G'}{X} = k + m - \binom{k}{2}$. On the other hand, suppose that there is a set of vertices $X$ of $G'$ such that $\card{X}\leq k'$ and $\vboundG{G'}{X}\leq t$.
 First, we note that $X$ cannot contain any vertices of $H_V$, as then $N_{G'}(X)$ would be too large. Thus, the set $X$ consists of vertices of $H_E$. Furthermore, we have that $\card{X} = \binom{k}{2}$. Indeed, assume that this is not the case. If $\card{X} \le \binom{k-1}{2} = \binom{k}{2} - k$, then, since $X$ has at least one neighbor in $H_V$, we have
\[\vboundG{G'}{X} \ge m - \card{X} + 1 \ge m - \binom{k}{2} + k + 1\,,\]
and if $\binom{k-1}{2} < \card{X} < \binom{k}{2}$, then $X$ has at least $k$ neighbors in $H_V$, and thus
\[\vboundG{G'}{X} \ge m - \card{X} + k > m - \binom{k}{2} + k\,.\]
Thus, we have that $X$ only consist of vertices of $H_E$ and $\card{X} = \binom{k}{2}$. But then the vertices of $H_V$ that are in $N_{G'}(X)$ form a $k$-clique in $G$.

The \classW{1}-hardness proof for \textsc{Cutting at Most  $k$ Vertices with Terminal} uses the same arguments. The only difference is that we add the terminal $s$ in the clique $H_E$ and let $k'=\binom{k}{2}+1$ (see Fig.~\ref{fig:W} b).

Because \textsc{Clique} is well known to be \classNP-complete~\cite{GareyJ79} and our parameterized reductions are polynomial in $k$, it immediately follows that \textsc{Cutting at Most  $k$ Vertices} and \textsc{Cutting at Most  $k$ Vertices with Terminal} are \classNP-complete.
\end{proof}

%PG:
While we have an \classFPT-algorithm for  \textsc{Cutting at Most  $k$ Vertices} when parameterized by $k$ and $t$ or by $t$ only, it is unlikely that the problem has a polynomial kernel (we refer to~\cite{DowneyF99,FlumG06,Niedermeierbook06} for the formal definitions of kernels). Let $G$ be a graph with $s$ connected components $G_1,\ldots,G_s$, and let $k\geq 1$, $t\geq 0$ be integers. Now $(G,k,t)$ is a YES-instance of  \textsc{Cutting at Most  $k$ Vertices} if and only if $(G_i,k,t)$ is a YES-instance for some $i\in\{1,\ldots,s\}$, because 
it can always be assumed that a solution is connected.
By the results of Bodlaender et al.~\cite{BodlaenderDFH09}, this together with Theorem~\ref{thm:svc_is_hard} implies the following.

\begin{theorem}\label{prop:kernel}
 \textsc{Cutting at Most  $k$ Vertices} has no polynomial kernel when parameterized either by $k$ and $t$ or by $t$ only, unless $\classNP\subseteq\text{\rm coNP}/\text{\rm poly}$.
\end{theorem}

\begin{figure}[ht]
\centering\scalebox{0.6}{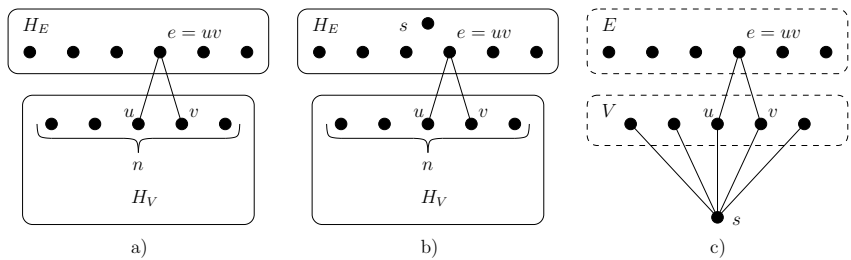}
\caption{Constructions of $G'$ in the proofs of Theorems~\ref{thm:svc_is_hard} and \ref{thm:svc_terminal_hard}.
\label{fig:W}}
\end{figure}

We will next show that when we consider the size of the separator $t$ as the sole parameter, adding a terminal makes the problem harder. Indeed, while \textsc{Cutting at Most  $k$ Vertices with Terminal} with parameter $t$ is trivially in XP, we next show that it is also \classW{1}-hard, in contrast to Theorem~\ref{thm:svc_fpt}.

\begin{theorem}\label{thm:svc_terminal_hard}
\textsc{Cutting at Most  $k$ Vertices with Terminal} is \classW{1}-hard with parameter $t$.
\end{theorem}

\begin{proof}
Again, we prove the claim by a reduction from \textsc{Clique}. Let $(G,k)$ be a clique instance, $n=|V(G)|$ and $m=|E(G)|$; we create an instance $(G', k', t, s)$ of \textsc{Cutting at Most  $k$ Vertices with Terminal}. The graph $G'$ is constructed as follows. Create a new vertex $s$ as the terminal. For each vertex and edge of $G$, add a corresponding vertex to $G'$, and add an edge between vertices $v$ and $e$ in $G'$ when $e$ is incident to $v$ in $G$. Finally, connect all vertices of $G'$ corresponding to vertices of $G$ to the terminal $s$, and set $k' = n - k + m - \binom{k}{2} + 1$ and $t = k$. The construction is shown in Fig.~\ref{fig:W} c).

If $G$ has a $k$-clique $K$, then cutting away the $k$ vertices of $G'$ corresponding to $K$ leaves exactly $n - k + m - \binom{k}{2}+1$ vertices in the connected component of $G'-K$ containing $s$. Now suppose that $X \subseteq V(G')$ is a set with $s \in X$ such that $\card{X} \le k'$ and $\vboundG{G'}{X} \le t$, and let $S = N_{G'}(X)$. Note that the elements of $V(G')$ that do not belong to $X$ are exactly the elements $v \in S$ and the vertices corresponding to $e = uv$ such that $u, v \in S$ and $e \notin S$; denote this latter set of elements by $E_0$. Since $X$ is a solution, we have $\card{S} + \card{E_0} \ge \binom{k}{2} + k$, and thus $\card{E_0} \ge \binom{k}{2}$. But this is only possible if $S$ is a $k$-clique in $G$.
\end{proof}

Finally, we show that \textsc{Cutting at most  $k$ vertices by edge-cut with terminal} is also \classNP-hard.

\begin{theorem}\label{thm:sec_terminal_np}
\textsc{Cutting at most  $k$ vertices by edge-cut with terminal} is \classNP-complete.
\end{theorem}

\begin{proof}
We give a reduction from the \textsc{Clique} problem. It is known that this problem is \classNP-complete for regular graphs~\cite{GareyJ79}. Let $(G,k)$ be an instance of \textsc{Clique}, with $G$ being a $d$-regular $n$-vertex graph. We create an instance $(G', k', t, s)$ of \textsc{Cutting $k$ Vertices by Edge-Cut with Terminal} as follows. The graph $G'$ is constructed by starting from a base clique of size $d n$. One vertex in this base clique is selected as the terminal $s$, and we additionally distinguish $d$ special vertices. For each $v \in V(G)$, we add a new vertex to $G'$, and add an edge between this vertex and all of the $d$ distinguished vertices of the base clique. For each edge $ e = uv$ in $G$, we also add a new vertex to $G'$, and add edges between this vertex and vertices corresponding to $u$ and $v$.
The construction is shown in Fig.~\ref{fig:NPc}.
We set $k' = dn + k + \binom{k}{2}$ and $t = d n - 2 \binom{k}{2}$.

\begin{figure}[ht]
\centering\scalebox{0.6}{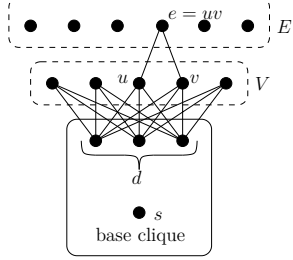}
\caption{Construction of $G'$ in the proof of Theorem~\ref{thm:sec_terminal_np}.
\label{fig:NPc}}
\end{figure}

If $G$ has a $k$-clique $K$, then selecting as $X$ the base clique and all vertices of $G'$ corresponding to vertices and edges of $K$ gives a solution to $(G',k',t,s)$, as we have $\card{X} = dn + k + \binom{k}{2}$ and 
$\ebound{X} = (d n - d k) + \bigl(d k - 2 \binom{k}{2}\bigr) = d n - 2 \binom{k}{2}$. For the other direction, consider any solution $X$ to instance $(G',k',t,s)$. The set $X$ must contain the whole base clique, as otherwise there are at least $dn - 1$ edges inside the base clique that belong to $\partial(X)$. Let $V_0 \subseteq V$ and $E_0 \subseteq E$ be the subsets of $X$ corresponding to vertices and edges of $G$, respectively. If $E_0 = \emptyset$, then $\ebound{X} = d n$. Assume now that $V_0$ is fixed, and consider how adding vertices to $E_0$ changes  $\ebound{X}$. For each edge $e \in E(G)$, if neither of the endpoints of $e$ is in $V_0$, then adding $e$ to $E_0$ adds $2$ to $\ebound{X}$. If exactly one of the endpoints of $e$ is in $V_0$, then adding $e$ to $E_0$ does not change $\ebound{X}$. Finally, if both of the endpoints of $e$ are in $V_0$, then adding $e$ to $E_0$ reduces $\ebound{X}$ by $2$. Thus, in order to have $\ebound{X} \le d n  - 2 \binom{k}{2}$, we must have that $\card{E_0} \ge \binom{k}{2}$ and the endpoints of all edges in $E_0$ are in $V_0$. But due to the requirement that $\card{X} \le  dn + k + \binom{k}{2}$, this is only possible if $V_0$ induces a clique in $G$.
\end{proof}

\paragraph{Acknowledgements.} We thank the anonymous reviewers for pointing out the related work of Lokshtanov and Marx.

\bibliographystyle{amsplain}
\bibliography{refs}

\end{document}